\documentclass[journal]{IEEEtran}
\usepackage{amsmath}
\usepackage{cite}
\usepackage{graphicx}
\usepackage{epstopdf}
\usepackage{amsfonts,amsmath,amssymb}
\usepackage{graphicx}
\usepackage{url}
\usepackage{bm}
\usepackage{bbm}
\usepackage{subfigure}
\usepackage{stfloats}
\usepackage{citesort}
\usepackage{algorithmic}
\usepackage{algorithm}
\usepackage{color}

\newcommand{\Hnull}{\mathcal{H}_0}
\newcommand{\Halt}{\mathcal{H}_1}

\newcommand{\Honull}{\mathcal{{H}}_0}
\newcommand{\Hoalt}{\mathcal{{H}}_1}
\newtheorem{theorem}{\noindent \textit{Theorem}}
\newtheorem{lemma}{\noindent \textit{Lemma}}
\newtheorem{corollary}{Corollary}

\DeclareMathOperator*{\argmin}{argmin}

\begin{document}

\title{Location Verification Systems Under Spatially Correlated Shadowing}

\author{Shihao~Yan, Ido~Nevat, Gareth~W.~Peters, and Robert~Malaney

\thanks{S. Yan and R. Malaney are with the School of Electrical Engineering and Telecommunications, The University of New South Wales, Sydney, NSW 2052, Australia (email: shihao.yan@unsw.edu.au; r.malaney@unsw.edu.au).}
\thanks{I. Nevat is with Institute for Infocomm Research, A$^{\star}$STAR, Singapore (email: ido-nevat@i2r.a-star.edu.sg).}
\thanks{G. W. Peters is with the Department of Statistical Science, University College London, London, United Kingdom (email: gareth.peters@ucl.ac.uk).}
\thanks{Part of this work has been presented in IEEE ICC 2014 \cite{yan2014signal}.}}



\vspace{-3cm}

\maketitle

\begin{abstract}
The verification of the location information utilized in wireless communication networks is a subject of growing importance.  In this work we formally analyze, for the first time, the  performance of a wireless Location Verification System (LVS)  under the realistic  setting of spatially correlated shadowing. Our analysis illustrates that anticipated levels of correlated shadowing can lead to a dramatic performance improvement of a Received Signal Strength (RSS)-based LVS.  We also analyze the performance of an LVS that utilizes Differential Received Signal Strength (DRSS), formally proving the rather counter-intuitive result that a DRSS-based LVS has identical performance to that of an RSS-based LVS, for all levels of correlated shadowing. Even more surprisingly, the identical performance of RSS and DRSS-based LVSs is found to hold even when the adversary does not optimize his true location. Only in the case where the adversary does not optimize \emph{all} variables under her control, do we find the performance of an RSS-based LVS to be better than a DRSS-based LVS. The results reported here are important for a wide range of emerging wireless communication applications whose proper functioning depends on the authenticity of the location information reported by a transceiver.
\end{abstract}

\begin{keywords}
Location verification, wireless networks, Received Signal Strength (RSS), Differential Received Signal Strength (DRSS), spatially correlated shadowing.
\end{keywords}

\IEEEpeerreviewmaketitle

\section{Introduction}

As location information becomes of growing importance in wireless networks, procedures to formally authenticate (verify) that information has attracted considerable research interest \cite{malaney2004location,vora2006secure,malaney2007securing,capkun2008secure,sheng2008detecting,chen2010detecting,zekavat2011handbook,chiang2012secure,yan2012information,yan2014optimal}. In a wide range of emerging wireless networks, the system may request a device (user) to report a location  obtained through some independent means (\emph{e.g.}, via a Global Positioning System (GPS) receiver embedded in the device). Such location information can be used to empower some functionality of the wireless network such as in  geographic routing protocols (\emph{e.g.}, \cite{leinmuller2005influence,leinmuller2006greedy,rabayah2012anew}), to provide for location-based access control protocols (\emph{e.g.}, \cite{chen2006inverting,capkun2010integrity}) or to provide some new location-based services (\emph{e.g.}, location-based key generation \cite{liu2008lke}). However, the use of location information as an enabler of functionality or services within the wireless network, also provides ample opportunity to attack the system since any reported location information (such as GPS) can be easily spoofed.
Such potential attacks are perhaps most concerning in the context of emerging Intelligent Transport Systems (ITS) such as wireless vehicular networks, where spoofed positions may lead to catastrophic results for vehicular collision-avoidance systems \cite{yan2009providing}.


In this work, we focus on a formal analysis of LVSs that attempt to verify a user's claimed location (such as a GPS location) based on independent observations received by the the wireless communications network itself. The inference in such an LVS is carried out to determine whether the claimed location represents a \textit{legitimate user} (a user who reports/claims to the network a location \emph{consistent} with his true position) or a \textit{malicious user} (a user who reports to the network a location \emph{inconsistent} with his true position).
A key difference between an LVS and a localization system is that the output of an LVS is a binary decision (legitimate/malicious user), whereas in localization system the output is an estimated location \emph{e.g.}, \cite{wang2011anew,gu2011localization,montorsi2014map}.
As such, an LVS is  provided with some additional \emph{a priori} (but potentially false) location information (\emph{i.e.}, a claimed location).

Since the RSS measured by wireless network is easily obtained, many location verification algorithms that utilize RSS as input observations have been developed (\emph{e.g.}, \cite{malaney2007securing,sheng2008detecting,chen2010detecting,yan2012information,yan2014optimal}). In addition, RSS can be readily combined with other location information metrics in order to improve the performance of a localization system \cite{malaney2007nuisance,wang2013cramer}. However, shadowing is one of the most influential factors in RSS-based LVSs, and all existing studies in RSS-based LVSs have made a simplified but unrealistic assumption that the shadowing at two different locations is uncorrelated. As per many empirical studies, the shadowing at different locations will be significantly correlated when the locations are close to each other or different locations possess similar terrain configurations \emph{e.g.}, \cite{gudmundson1991correlation,liberti1992statistics,zayana1998measurements}.
Although some specific  studies  have investigated the performance of RSS-based localization systems under correlated shadowing \cite{patwari2008effects,patwari2009correlated,vaghefi2013received}, the impact of spatially correlated shadowing on RSS-based LVSs under realistic threat models has not been previously explored. This leaves an important gap in our understanding on the performance levels of RSS-based LVSs in realistic wireless channel settings and under realistic threat models. The main purpose of this paper is to close this gap.

Further to our considerations of RSS-based LVSs, we note that there could be circumstances when  use of Differential Received Signal Strength (DRSS) in the LVS context may be beneficial. Indeed it is well known that there are a range of scenarios in which the use of DRSS is more suitable for wireless location acquisition \cite{wang2013robust}. One example is where users do not have a common transmit power setting on all devices. However, the performance of DRSS-based LVSs have not yet been analyzed in the literature. This work also closes this gap, extending our analysis of DRSS-based LVSs to the correlated shadowing regime. This will allow us to provide a detailed performance comparison between RSS-based LVSs and DRSS-based LVSs  under correlated shadowing - a comparison that provides for a few surprising results.

A summary of the main contributions of this work are as follows. (i) Under spatially correlated $\log$-normal shadowing, we analyze the detection performance of an RSS-based LVS in terms of false positive and detection rates. Our analysis demonstrates that the spatial correlation of the shadowing leads to a significant performance improvement for the RSS-based LVS relative to the case with uncorrelated shadowing (a doubling of the detection rate for a given false positive rate for anticipated correlation levels). (ii) We analyze the detection performance of a DRSS-based LVS under spatially correlated shadowing, proving that the detection performance of the DRSS-based LVS is identical to that of the RSS-based LVS. As we discuss later, this result is rather surprising. (iii) We analyze our systems under a relaxed threat model scenario in which the adversary whose actual location is physically constrained (\emph{e.g.}, constrained within a building) and therefore cannot optimize his location for the attack. We show that even in these circumstance the performance of the RSS-based LVS and the DRSS-based LVS remain identical. (iv) Finally, we illustrate the case where the RSS-based LVS do have advantages over the DRSS-Based LVS, namely, when the adversary does not (or cannot) optimize his boosted transmit power level.

The rest of this paper is organized as follows. Section \ref{sec_system} details our system model. In Section \ref{sec_RSS}, the detection performance of an RSS-based LVS is analyzed under spatially correlated shadowing. In Section \ref{sec_DRSS}, the detection performance of a DRSS-based LVS is analyzed, and a throughout performance comparison between the RSS-based LVS and the DRSS-based LVS is provided. Section \ref{sec_numerical} provides numerical results to verify the accuracy of our analysis. Finally, Section \ref{sec_conclusion} draws concluding remarks.

\section{System Model}\label{sec_system}

\subsection{Assumptions}

We outline the system model and state the assumptions adopted in this work.

\begin{enumerate}
\item {A \emph{single} user (legitimate or malicious) reports his claimed location, $\mathbf{x}_c = [x^1_c, x^2_c] \in \mathbb{R}^2$, to a network with $N$ Base Stations (BSs) in the communication range of the user, where the publicly known location of the $i$-th BS is $\mathbf{x}_i = [x^1_i,x^2_i] \in \mathbb{R}^2$ ($i = 1, 2, \dots, N$). One of the $N$ BSs is the Process Center (PC), and all other BSs will transmit the measurements collected from the user to the PC. The PC is to make decisions based on the user's claimed location and the measurements collected by all the $N$ BSs.}

\item {A user (legitimate or malicious) can obtain his true position, $\mathbf{x}_t = [x^1_t,x^2_t]$, from his localization equipment (\emph{e.g.}, GPS), and that the localization error is zero. Thus, a legitimate user's claimed location, $\mathbf{x}_c$, is exactly the same as his true location. However, a malicious user will falsify (spoof) his claimed position in an attempt to fool the LVS. We assume the spoofed claimed location of the malicious user is also $\mathbf{x}_c$.}

\item We adopt the minimum distance model as our threat model, in which the distance between the malicious user's true location and his claimed location is greater or equal to $r$, \emph{i.e.}, $|\mathbf{x}_c -\mathbf{x}_t| \geq r$.

\item {We denote the null hypothesis where the user is legitimate as $\Hnull$, and denote the alternative hypothesis where the user is malicious as $\Halt$. The \emph{a priori} knowledge at the LVS can be summarized as
\begin{eqnarray}\label{prior}
 \left\{ \begin{aligned}\label{ncon}
        \ & \Hnull:~\mathbf{x}_c = \mathbf{x}_t \;\;\text{(legitimate user)}\\
        \ & \Halt:~|\mathbf{x}_c-\mathbf{x}_t| \geq r \;\;\text{(malicious user)}.
         \end{aligned} \right.
\end{eqnarray}}
\end{enumerate}

\subsection{Observation Model under $\Hnull$ (legitimate user)}

Based on the $\log$-normal propagation model, the RSS (in dB) received by the $i$-th BS from a legitimate user, $y_i$, is given by
\begin{eqnarray}\label{observ_H0}
y_i = u_i + \omega_i, ~~i = 1, 2, \dots, N,
\end{eqnarray}
where
\begin{eqnarray}\label{definition_u}
u_i = p - 10\gamma \log_{10}\left(\frac{d_i^c}{d}\right),
\end{eqnarray}
and $p$ is a reference received power corresponding to a reference distance $d$, $\gamma$ is the path loss exponent, $\omega_i$ is a zero-mean normal random variable with variance $\sigma_{dB}^2$, and $d_i^c$ is the Euclidean distance from the $i$-th BS to the legitimate user's claimed location (also his true location) given by $d_i^c=|\mathbf{x}_c-\mathbf{x}_i|$. In practice, in order to determine the values of a pair of $p$ and $d$ we have to know the transmit power of a legitimate user. Under spatially correlated shadowing, $\omega_i$ is correlated to $\omega_j$ ($j = 1,2,\dots,N$), and the $N \times N$ covariance matrix of $\bm{\sigma} = [\sigma_1, \dots, \omega_N]^T$ is denoted as $\mathbf{R}$. Adopting the well-known spatially correlated shadowing model of \cite{zekavat2011handbook,gudmundson1991correlation}, the $(i,j)$-th element of $\mathbf{R}$ is given by
\begin{equation}\label{shadowing}
R_{ij} = \sigma_{dB}^2 \exp\left(-\frac{d_{ij}}{D_c}\ln 2\right), ~~j = 1, 2, \dots, N,
\end{equation}
where $d_{ij} = ||\mathbf{x}_i-\mathbf{x}_j||_2$ is the Euclidean distance from the $i$-th BS to the $j$-th BS, and $D_c$ is
a constant in units of distance, at which the correlation coefficient reduces to $1/2$ (in this work all distances are in meters).
From \eqref{shadowing}, we can see that the correlation between $\omega_i$ and $\omega_j$ decreases as $d_{ij}$
increases ($R_{ij} = \sigma_{dB}^2$ when $i=j$, and $R_{ij} \rightarrow 0$ as $d_{ij} \rightarrow \infty$). We also note that $R_{ij}$
increases as $D_c$ increases for a given $d_{ij}$. As such, $D_c$ is a parameter that indicates the degree of shadowing correlation in
some specific environment (for a given $d_{ij}$, a larger $D_c$ means that the shadowing is more correlated).

Based on \eqref{observ_H0}, we can see that under $\Hnull$ the $N$-dimensional observation vector $\mathbf{y} = [y_1, \dots, y_N]^T$ follows a multivariate normal distribution, which is
\begin{equation}\label{likelihood_H0}
f\left(\mathbf{y}|\Hnull\right) = \mathcal{N} \left(\mathbf{u}, \mathbf{R}\right),
\end{equation}
where $\mathbf{u} = [u_1, u_2, \dots, u_N]^T$ is the mean vector.
\subsection{Observation Model under $\Halt$ (malicious user)}

In practice, in addition to spoofing the claimed location, the malicious user can also adjust his transmit power to impact the RSS values received by all BSs in order to minimize the probability of being detected. As such, the RSS received by the $i$-th BS from a malicious user, $y_i$, is given by
\begin{equation}\label{threat}
y_i = p_x + v_i + \omega_i,
\end{equation}
where
\begin{eqnarray}\label{definition_v}
v_i = p - 10\gamma \log_{10}\left(\frac{d_i^t}{d}\right),
\end{eqnarray}
$d_i^t$ is the Euclidean distance from the $i$-th BS to the malicious user's true location given by $d_i^t=||\mathbf{x}_t-\mathbf{x}_i||_2$, and $p_x$ is the additional boosted transmit power.
Based on \eqref{threat}, under $\Halt$ the $N$-dimensional observation vector $\mathbf{y}$, conditioned on known $p_x$ and $\mathbf{x}_t$, also follows a multivariate normal distribution, which is
\begin{equation}\label{likelihood_H1}
f\left(\mathbf{y}|p_x, \mathbf{x}_t, \Halt\right) = \mathcal{N} \left(p_x \mathbf{1}_N + \mathbf{v}, \mathbf{R}\right),
\end{equation}
where $\mathbf{1_N}$ is a $N \times 1$ vector with all elements set to unity and $\mathbf{v} = [v_1, v_2, \dots, v_N]^T$. We note that in practice $p_x$ and $\mathbf{x}_t$ are set by the malicious user.


\subsection{Decision Rule of an LVS}

We adopt the Likelihood Ratio Test (LRT) as the decision rule since it is known that the LRT achieves the highest detection rate for any given false positive rate \cite{neyman1933problem}. Therefore, the LRT can achieve the minimum Bayesain average cost and the maximum mutual information between the input and output of an LVS \cite{yan2014optimal}. The LRT decision rule is given by
\begin{equation}\label{arbitrary}
\Lambda\left(\psi (\mathbf{y})\right) \triangleq \frac{f\left(\psi (\mathbf{y})|\Halt\right)}{f\left(\psi (\mathbf{y})|\Hnull\right)} \begin{array}{c}
\overset{\Hoalt}{\geq} \\
\underset{\Honull}{<}
\end{array}%
\lambda,
\end{equation}
where $\Lambda\left(\psi (\mathbf{y})\right)$ is the test statistic, $\psi (\mathbf{y})$ is a predefined transformation of the observations $\mathbf{y}$ (to be determined in a specific LVS, \emph{e.g.}, RSS or DRSS),
$f\left(\psi (\mathbf{y})|\Halt\right)$ is the marginal likelihood (probability density function of $\psi (\mathbf{y})$) under $\Halt$, $f\left(\psi (\mathbf{y})|\Hnull\right)$ is the marginal likelihood under $\Hnull$, $\lambda$ is the threshold corresponding to $\Lambda\left(\psi (\mathbf{y})\right)$, $\Honull$ and $\Hoalt$ are the binary decisions that infer whether the user is legitimate or malicious, respectively. Given the decision rule in \eqref{arbitrary}, the false positive and detection rates of an LVS are functions of $\lambda$. The specific value of $\lambda$ can be set through minimizing the Bayesian average cost or maximizing the mutual information between the system input and output in the information-theoretic framework. The intrinsic core performance metrics of an LVS are false positive and detection rates, other potential performance metrics can be written as functions of these two rates. As such, in this work we adopt the false positive and detection rates as the performance metrics for an LVS.

\section{RSS-based Location Verification System}\label{sec_RSS}

In this section, we analyze the performance of the RSS-based LVS in terms of the false positive and detection rates, based on which we examine the impact of the spatially correlated shadowing.

\subsection{Attack Strategy of the Malicious User}

We assume that the malicious user optimizes all the parameters under his control. This assumption is adopted in most threat models. The malicious user will therefore optimize his $p_x$ and $\mathbf{x}_t$ such that the difference between $f\left(\mathbf{y}|\Hnull\right)$ and $f\left(\mathbf{y}|p_x, \mathbf{x}_t, \Halt\right)$ is minimized in order to minimize the probability to be detected. Here, we adopt the Kullback-Leibler (KL) divergence to quantify the difference between $f\left(\mathbf{y}|\Hnull\right)$ and $f\left(\mathbf{y}|p_x, \mathbf{x}_t, \Halt\right)$, which is a measure of the information loss when $f\left(\mathbf{y}|p_x, \mathbf{x}_t, \Halt\right)$ is used to approximate $f\left(\mathbf{y}|\Hnull\right)$ \cite{kullback1951on}.

Based on \eqref{likelihood_H0} and \eqref{likelihood_H1}, the KL divergence between $f\left(\mathbf{y}|\Hnull\right)$ and $f\left(\mathbf{y}|p_x, \mathbf{x}_t, \Halt\right)$ is given by
\begin{equation}\label{KL_divergence}
\begin{split}
\phi(p_x, \mathbf{x}_t) &= D_{KL}\left[f\left(\mathbf{y}| \Hnull\right)||f\left(\mathbf{y}|p_x, \mathbf{x}_t \Halt\right)\right]\\
&= \int_{-\infty}^{\infty} \ln \frac{f\left(\mathbf{y}|\Hnull\right)}{f\left(\mathbf{y}|p_x, \mathbf{x}_t, \Halt\right)} f\left(\mathbf{y}|\Hnull\right) d{\mathbf{y}}\\
&= \frac{1}{2}(p_x \mathbf{1}_N + \mathbf{v} - \mathbf{u})^T \mathbf{R}^{-1}(p_x \mathbf{1}_N + \mathbf{v} - \mathbf{u}).
\end{split}
\end{equation}
Then, the optimal values of $p_x$ and $\mathbf{x}_t$ that minimize $\phi(p_x, \mathbf{x}_t)$ can be obtained through
\begin{align}\label{optimal_two}
(p_x^{\ast}, \mathbf{x}_t^{\ast}) = \argmin_{p_x, ||\mathbf{x}_t-\mathbf{x}_c||_2 \geq r} \phi(p_x, \mathbf{x}_t).
\end{align}
The closed-form expressions for $p_x^{\ast}$ and $\mathbf{x}_t^{\ast}$ are intractable, but they can be obtained through numerical search. In order to simplify the numerical search, we first derive the optimal value of $p_x$ for a given $\mathbf{x}_t$, which is presented in the following lemma.

\begin{lemma}\label{lemma1}
The optimal value of $p_x$ that minimizes $\phi(p_x, \mathbf{x}_t)$ for any given $\mathbf{x}_t$ is
\begin{equation}\label{optimal_px}
p_x^o(\mathbf{x}_t) = \frac{(\mathbf{u} - \mathbf{v})^T \mathbf{R}^{-1}\mathbf{1}_N}
{\mathbf{1}_N^T\mathbf{R}^{-1}\mathbf{1}_N}.
\end{equation}
\end{lemma}

\begin{IEEEproof}
The first derivative of $\phi(p_x, \mathbf{x}_t)$ with respect to $p_x$ is derived as
\begin{equation}\label{first_px}
\begin{split}
\frac{\partial \phi(p_x, \mathbf{x}_t)}{\partial p_x} &= \frac{\partial \phi(p_x, \mathbf{x}_t)}{\partial \left(p_x \mathbf{1}_N\right)}\frac{\partial \left(p_x \mathbf{1}_N\right)}{\partial p_x}\\
&=(p_x \mathbf{1}_N + \mathbf{v} - \mathbf{u})^T \mathbf{R}^{-1}\frac{\partial \left(p_x \mathbf{1}_N\right)}{\partial p_x}\\
&=(p_x \mathbf{1}_N + \mathbf{v} - \mathbf{u})^T \mathbf{R}^{-1}\mathbf{1}_N.
\end{split}
\end{equation}
Following \eqref{first_px}, the second derivative of $\phi(p_x, \mathbf{x}_t)$ with respect to $p_x$ is derived as
\begin{align}\label{second_px}
\frac{\partial^2 \phi(p_x, \mathbf{x}_t)}{\partial^2 p_x} = \mathbf{1}_N^T \mathbf{R}^{-1}\mathbf{1}_N.
\end{align}
Since $\mathbf{R}$ is a positive-definite symmetric matrix, as per \eqref{second_px} we have ${\partial^2 \phi(p_x, \mathbf{x}_t)}/{\partial^2 p_x} >0$, which indicates that $\phi(p_x, \mathbf{x}_t)$ is a convex function of $p_x$. As such, setting ${\partial \phi(p_x, \mathbf{x}_t)}/{\partial p_x}=0$, we obtain the desired result in \eqref{optimal_px} after some algebraic manipulations.
\end{IEEEproof}

From Lemma~\ref{lemma1}, we note that the malicious user optimizes his transmit power, \emph{i.e.}, $p_x = p_x^o(\mathbf{x}_t)$, to compensate the path-loss difference between his claimed location and his true location. We also note that $p_x^o(\mathbf{x}_t)$ is a function of $\mathbf{R}$ under spatial correlated shadowing. This is different from the scenario with uncorrelated shadowing, where $p_x^o(\mathbf{x}_t)$ is independent of the shadowing noise \cite{yan2014optimal}. Substituting $p_x^o(\mathbf{x}_t)$ into \eqref{KL_divergence}, we have
\begin{equation}\label{KL_divergence_min}
\begin{split}
\phi(p_x^o(\mathbf{x}_t), \mathbf{x}_t) &= \frac{1}{2}(\mathbf{w}-\mathbf{u})^T \mathbf{R}^{-1}(\mathbf{w}-\mathbf{u}),
\end{split}
\end{equation}
where
\begin{align}\label{definition_w}
\mathbf{w} = \frac{(\mathbf{u} - \mathbf{v})^T \mathbf{R}^{-1}\mathbf{1}_N}
{\mathbf{1}_N^T\mathbf{R}^{-1}\mathbf{1}_N}\mathbf{1}_N + \mathbf{v}.
\end{align}
Since we have shown that $\phi(p_x, \mathbf{x}_t)$ is a convex function of $p_x$ in \eqref{second_px}, $\mathbf{x}_t^{\ast}$ is given by
\begin{align}\label{optimal_true}
\mathbf{x}_t^{\ast} = \argmin_{||\mathbf{x}_t-\mathbf{x}_c||_2 \geq r} \phi(p_x^o(\mathbf{x}_t), \mathbf{x}_t).
\end{align}
Substituting $\mathbf{x}_t^{\ast}$ into $p_x^o(\mathbf{x}_t)$, we obtain $p_x^{\ast} = p_x^o(\mathbf{x}_t^{\ast})$. We note that Lemma~1 is of importance since it reduces a three-dimension numerical search in \eqref{optimal_two} into a two-dimension numerical search in \eqref{optimal_true}.


Substituting $p_x^{\ast}$ and $\mathbf{x}_t^{\ast}$ into \eqref{threat}, the RSS received by the $i$-th BS from a malicious user can be written as
\begin{equation}\label{threat_general}
\mathbf{y} = \mathbf{w}^{\ast} + \bm{\omega},
\end{equation}
where
\begin{equation}\label{w_definition}
\mathbf{w}^{\ast} = \frac{(\mathbf{u} - \mathbf{v}^{\ast})^T \mathbf{R}^{-1}\mathbf{1}_N}
{\mathbf{1}_N^T\mathbf{R}^{-1}\mathbf{1}_N}\mathbf{1}_N + \mathbf{v}{\ast},
\end{equation}
$\mathbf{v}^{\ast}$ is obtained by substituting $\mathbf{x}_t^{\ast}$ into $\mathbf{v}$, and $\bm{\omega}~ =~[\omega_1, \omega_2, \dots, \omega_N]$.
Based on \eqref{threat_general}, the likelihood function under $\Halt$ conditioned on $p_x^{\ast}$ and $\mathbf{x}_t^{\ast}$ can be written as
\begin{equation}\label{likelihood_H1_FFA}
f\left(\mathbf{y}|p_x^{\ast}, \mathbf{x}_t^{\ast}, \Halt\right) = \mathcal{N} (\mathbf{w}^{\ast}, \mathbf{R}).
\end{equation}

\subsection{Performance of the RSS-based LVS}

In some practical cases, the malicious user may not have the freedom to optimize his true location, \emph{e.g.}, if the malicious user is physically limited to be inside a building. However, the malicious user can still optimize his transmit power as per his true location. As such, without losing generality, we first analyze the performance of the RSS-based LVS for $p_x = p_x^o(\mathbf{x}_t)$, and then present the performance of the RSS-based LVS for $p_x = p_x^{\ast}$ and $\mathbf{x}_t = \mathbf{x}_t^{\ast}$ as a special case.

Following \eqref{arbitrary}, the specific LRT decision rule of the RSS-based LVS for $p_x = p_x^o(\mathbf{x}_t)$ is given by
\begin{equation}\label{RSS_known_ori}
\Lambda^o\left(\mathbf{y}\right) \triangleq \frac{f\left(\mathbf{y}|p_x^o(\mathbf{x}_t), \mathbf{x}_t, \Halt\right)}{f\left(\mathbf{y}|\Hnull\right)}  \begin{array}{c}
\overset{\Hoalt}{\geq} \\
\underset{\Honull}{<}
\end{array}%
\lambda_R^o,
\end{equation}
where $\Lambda^o\left(\mathbf{y}\right)$ is the likelihood ratio of $\mathbf{y}$ for $p_x = p_x^o(\mathbf{x}_t)$, $f\left(\mathbf{y}|p_x^o(\mathbf{x}_t), \mathbf{x}_t, \Halt\right) = \mathcal{N} (\mathbf{w}, \mathbf{R})$, and $\lambda_R^o$ is a threshold for $\Lambda^o\left(\mathbf{y}\right)$. Substituting \eqref{likelihood_H0} and \eqref{likelihood_H1_FFA} into \eqref{RSS_known_ori}, we obtain $\Lambda^o\left(\mathbf{y}\right)$ in the $\log$ domain as
\begin{equation*}
\begin{split}
\ln \Lambda^o\left(\mathbf{y}\right) &\!=\! \frac{1}{2}(\mathbf{y}\!\!-\!\!\mathbf{u})^T \mathbf{R}^{\!-\!1}(\mathbf{y}\!\!-\!\!\mathbf{u})\!\!-\!\! \frac{1}{2}(\mathbf{y}\!\!-\!\!\mathbf{w})^T \mathbf{R}^{\!-\!1}(\mathbf{y}\!\!-\!\!\mathbf{w})\\
&= \left(\mathbf{w} \!\!-\!\! \mathbf{u}\right)^T \mathbf{R}^{-1} \mathbf{y} \!\!-\!\! \frac{1}{2}\left(\mathbf{w} \!\!-\!\! \mathbf{u}\right)^T \mathbf{R}^{-1} \left(\mathbf{w} \!+\! \mathbf{u}\right).
\end{split}
\end{equation*}
As such, for the theorem to follow, we can rewrite the decision rule in \eqref{RSS_known_ori} as the following format
\begin{equation}\label{decision_RSS}
\mathbb{T}(\mathbf{y}) \begin{array}{c}
\overset{\Hoalt}{\geq} \\
\underset{\Honull}{<}
\end{array}%
\Gamma_R,
\end{equation}
where $\mathbb{T}(\mathbf{y})$ is the test statistic given by
\begin{equation}\label{statistic_RSS}
\mathbb{T}(\mathbf{y}) \triangleq \left(\mathbf{w} - \mathbf{u}\right)^T \mathbf{R}^{-1} \mathbf{y},
\end{equation}
and $\Gamma_R$ is the threshold for $\mathbb{T}(\mathbf{y})$ given by
\begin{equation}\label{threshold_RSS}
\Gamma_R \triangleq \ln \lambda_R^o + \frac{1}{2}\left(\mathbf{w} - \mathbf{u}\right)^T \mathbf{R}^{-1} \left(\mathbf{w} + \mathbf{u}\right).
\end{equation}

We then derive the false positive rate, $\alpha_R^{o}$, and detection rate, $\beta_R^{o}$, of the RSS-based LVS for $p_x = p_x^o(\mathbf{x}_t)$ in the following theorem.

\begin{theorem}\label{theorem1}
\emph{For $p_x = p_x^o(\mathbf{x}_t)$, the false positive and detection rates of the RSS-based LVS are
\begin{align}
\alpha_R^{o}(\mathbf{x}_t) &= \mathcal{Q}\left[\frac{\Gamma_R - \left(\mathbf{w} \!\!-\!\! \mathbf{u}\right)^T \mathbf{R}^{\!-\!1} \mathbf{u}}{\sqrt{\left(\mathbf{w} \!\!-\!\! \mathbf{u}\right)^T \mathbf{R}^{\!-\!1} \left(\mathbf{w} \!\!-\!\! \mathbf{u}\right)}}\right]\notag \\
&= \mathcal{Q}\left[\frac{\ln \lambda_R^o + \frac{1}{2}\left(\mathbf{w} \!\!-\!\! \mathbf{u}\right)^T \mathbf{R}^{\!-\!1} \left(\mathbf{w}\!\!-\!\!\mathbf{u}\right)}{\sqrt{\left(\mathbf{w} \!\!-\!\! \mathbf{u}\right)^T \mathbf{R}^{\!-\!1} \left(\mathbf{w} \!\!-\!\! \mathbf{u}\right)}}\right],\label{alpha_R_xt} \\
\beta_R^{o}(\mathbf{x}_t) &= \mathcal{Q}\left[\frac{\Gamma_R - \left(\mathbf{w} \!\!-\!\! \mathbf{u}\right)^T \mathbf{R}^{\!-\!1} \mathbf{w}}{\sqrt{\left(\mathbf{w} \!\!-\!\! \mathbf{u}\right)^T \mathbf{R}^{\!-\!1} \left(\mathbf{w} \!\!-\!\! \mathbf{u}\right)}}\right]\notag \\
&= \mathcal{Q}\left[\frac{\ln \lambda_R^o - \frac{1}{2}\left(\mathbf{w} \!\!-\!\! \mathbf{u}\right)^T \mathbf{R}^{\!-\!1} \left(\mathbf{w}\!\!-\!\!\mathbf{u}\right)}{\sqrt{\left(\mathbf{w} \!\!-\!\! \mathbf{u}\right)^T \mathbf{R}^{\!-\!1} \left(\mathbf{w} \!\!-\!\! \mathbf{u}\right)}}\right],\label{beta_R_xt}
\end{align}
where $\mathcal{Q}[x] = \frac{1}{\sqrt{2 \pi}}\int_x^{\infty} \exp (-t^2/2) dt$.}
\end{theorem}
\begin{proof}
Using \eqref{statistic_RSS}, the distributions of $\mathbb{T}(\mathbf{y})$ under $\Hnull$ and $\Halt$ are derived as follows
\begin{align}
&\mathbb{T}(\mathbf{y}) |\Hnull \notag \\
&\sim
\mathcal{N}\left(\left(\mathbf{w} \!\!-\!\! \mathbf{u}\right)^T \mathbf{R}^{\!-\!1} \mathbf{u}, \left(\mathbf{w} \!\!-\!\! \mathbf{u}\right)^T \mathbf{R}^{\!-\!1} \mathbf{R} \left(\left(\mathbf{w} \!\!-\!\! \mathbf{u}\right)^T \mathbf{R}^{\!-\!1}\right)^T \right)\notag\\
&=\mathcal{N}\left(\left(\mathbf{w} \!\!-\!\! \mathbf{u}\right)^T \mathbf{R}^{\!-\!1} \mathbf{u}, \left(\mathbf{w} \!\!-\!\! \mathbf{u}\right)^T \mathbf{R}^{\!-\!1} \left(\mathbf{w} \!\!-\!\! \mathbf{u}\right)\right),\label{R_0}\\
&\mathbb{T}(\mathbf{y}) |\Halt \notag \\
&\sim
\mathcal{N}\left(\left(\mathbf{w} \!\!-\!\! \mathbf{u}\right)^T \mathbf{R}^{\!-\!1} \mathbf{w}, \left(\mathbf{w} \!\!-\!\! \mathbf{u}\right)^T \mathbf{R}^{\!-\!1} \mathbf{R} \left(\left(\mathbf{w} \!\!-\!\! \mathbf{u}\right)^T \mathbf{R}^{\!-\!1}\right)^T \right)\notag\\
&=\mathcal{N}\left(\left(\mathbf{w} \!\!-\!\! \mathbf{u}\right)^T \mathbf{R}^{\!-\!1} \mathbf{w}, \left(\mathbf{w} \!\!-\!\! \mathbf{u}\right)^T \mathbf{R}^{\!-\!1} \left(\mathbf{w} \!\!-\!\! \mathbf{u}\right)\right).\label{R_1}
\end{align}
As per the decision rule in \eqref{decision_RSS}, the false positive and detection rates are given by
\begin{align}
\alpha_R^{o}(\mathbf{x}_t) &\triangleq \Pr\left(\mathbb{T}(\mathbf{y}) \geq \Gamma_R |\Hnull \right),\label{false_R}\\
\beta_R^{o}(\mathbf{x}_t) &\triangleq \Pr\left(\mathbb{T}(\mathbf{y}) \geq \Gamma_R |\Halt \right).\label{detection_R}
\end{align}
Substituting \eqref{R_0} and \eqref{R_1} into \eqref{false_R} and \eqref{detection_R}, respectively, we obtain the results in \eqref{alpha_R} and \eqref{beta_R} after some algebraic manipulations.
\end{proof}

For $p_x = p_x^{\ast}$ and $\mathbf{x}_t = \mathbf{x}_t^{\ast}$, the LRT decision rule of the RSS-based LVS is given by
\begin{equation}\label{RSS_known_worst}
\Lambda^{\ast}\left(\mathbf{y}\right) \triangleq \frac{f\left(\mathbf{y}|p_x^{\ast}, \mathbf{x}_t^{\ast}, \Halt\right)}{f\left(\mathbf{y}|\Hnull\right)}  \begin{array}{c}
\overset{\Hoalt}{\geq} \\
\underset{\Honull}{<}
\end{array}%
\lambda_R^{\ast},
\end{equation}
where $\Lambda^{\ast}\left(\mathbf{y}\right)$ is the likelihood ratio of $\mathbf{y}$ for $p_x = p_x^{\ast}$ and $\mathbf{x}_t = \mathbf{x}_t^{\ast}$ and $\lambda_R^{\ast}$ is a threshold for $\Lambda^{\ast}\left(\mathbf{y}\right)$.
Following Theorem~1, the false positive and detection rates of the RSS-based LVS for $p_x = p_x^{\ast}$ and $\mathbf{x}_t = \mathbf{x}_t^{\ast}$ are given by
\begin{align}
\alpha_R^{\ast}&= \mathcal{Q}\left[\frac{\ln \lambda_R^{\ast} + \frac{1}{2}\left(\mathbf{w}^{\ast} \!\!-\!\! \mathbf{u}\right)^T \mathbf{R}^{\!-\!1} \left(\mathbf{w}^{\ast}\!\!-\!\!\mathbf{u}\right)}{\sqrt{\left(\mathbf{w}^{\ast} \!\!-\!\! \mathbf{u}\right)^T \mathbf{R}^{\!-\!1} \left(\mathbf{w}^{\ast} \!\!-\!\! \mathbf{u}\right)}}\right],\label{alpha_R} \\
\beta_R^{\ast}&= \mathcal{Q}\left[\frac{\ln \lambda_R^{\ast} - \frac{1}{2}\left(\mathbf{w}^{\ast} \!\!-\!\! \mathbf{u}\right)^T \mathbf{R}^{\!-\!1} \left(\mathbf{w}^{\ast}\!\!-\!\!\mathbf{u}\right)}{\sqrt{\left(\mathbf{w}^{\ast} \!\!-\!\! \mathbf{u}\right)^T \mathbf{R}^{\!-\!1} \left(\mathbf{w}^{\ast} \!\!-\!\! \mathbf{u}\right)}}\right].\label{beta_R}
\end{align}

We note that the results provided in \eqref{alpha_R_xt} and \eqref{beta_R_xt} are based on an arbitrary true location $\mathbf{x}_t$ of the malicious user, which are more general than that provided in \eqref{alpha_R} and \eqref{beta_R}. That is, $\alpha_R^{\ast} = \alpha_R^o(\mathbf{x}_t^{\ast})$ and $\beta_R^{\ast} = \beta_R^o(\mathbf{x}_t^{\ast})$. By using \eqref{alpha_R_xt} and \eqref{beta_R_xt}, we can compare the performance of the RSS-based LVS with that of the DRSS-based LVS in a general scenario.

\section{DRSS-based Location Verification System}\label{sec_DRSS}

In this section, we analyze the detection performance of the DRSS-based LVS under spatially correlated shadowing. We also provide an analytical comparison between the RSS-based LVS and the DRSS-based LVS.

\subsection{DRSS Observations}\label{procedure_diff}

We obtain $(N-1)$ basic DRSS observations from $N$ RSS observations by subtracting the $N$-th RSS observation from all other $(N-1)$ RSS observations. As such, the $m$-th DRSS value under $\Hnull$ is given by
\begin{equation}\label{DRSS_observation0}
\Delta y_m = \Delta u_m + \Delta \omega_m, ~~m = 1, 2, \dots, N-1,
\end{equation}
where $\Delta u_m = u_m - u_N$, and $\Delta \omega_m = \omega_m - \omega_N$. We note that $\Delta \omega_m$ is Gaussian with zero mean and variance $2(\sigma_{dB}^2-R_{mN})$. We denote the $(N-1)\times(N-1)$ covariance matrix of the $(N-1)$-dimensional DRSS vector $\mathbf{\Delta y} = [\Delta y_1, \dots, \Delta y_{N-1}]^T$ as $\mathbf{D}$, whose $(m,n)$-th element is given by ($n = 1,2,\dots,N-1$)
\begin{align}\label{D_definition}
D_{mn}  = R_{NN} + R_{mn} - R_{mN}-R_{nN}.
\end{align}
As such, $\mathbf{\Delta y}$ under $\Hnull$ follows a multivariate normal distribution, which is given by
\begin{equation}\label{likelihood_DRSS_0}
f\left(\mathbf{\Delta y}|\Hnull\right) = \mathcal{N} (\mathbf{\Delta u}, \mathbf{D}),
\end{equation}
where $\mathbf{\Delta u} = [\Delta u_1, \dots, \Delta u_{N-1}]^T$ is the mean vector.

Likewise, the $m$-th DRSS value under $\Halt$ is
\begin{equation}\label{DRSS_observation1}
\Delta y_m = \Delta v_m + \Delta \omega_m,
\end{equation}
where $\Delta v_m = v_m - v_N$. Noting $\mathbf{\Delta v} = [\Delta v_1, \dots, \Delta v_{N-1}]^T$, $\mathbf{\Delta y}$ under $\Halt$ follows another multivariate normal distribution, which is given by
\begin{equation}\label{likelihood_DRSS_1}
f\left(\mathbf{\Delta y}|\mathbf{x}_t, \Halt\right) = \mathcal{N} (\mathbf{\Delta v}, \mathbf{D}).
\end{equation}

\subsection{Attack Strategy of the Malicious User}

As per \eqref{definition_u} and \eqref{definition_v}, we know that both $p$ and $d$ are constant at all elements of $\mathbf{u}$ and $\mathbf{v}$. As such, based on \eqref{DRSS_observation0} and \eqref{DRSS_observation1} we can see that $\mathbf{\Delta y}$ under both $\Hnull$ and $\Halt$ are independent of $p$ and $d$, and therefore both $f\left(\mathbf{\Delta y}|\Hnull\right)$ and $f\left(\mathbf{\Delta y}|\mathbf{x}_t, \Halt\right)$ are independent of $p$ and $d$. Therefore, in the DRSS-based LVS the malicious user does not need to adjust his transmit power in order to minimize the probability to be detected. In the DRSS-based LVS, the malicious user only has to optimize his true location through minimizing the KL-divergence between $f\left(\mathbf{\Delta y}|\Hnull\right)$ and $f\left(\mathbf{\Delta y}|\mathbf{x}_t, \Halt\right)$, which is given by
\begin{equation}\label{KL_divergence_DRSS}
\begin{split}
\varphi(\mathbf{x}_t) &= D_{KL}\left[f\left(\mathbf{\Delta y}| \Hnull\right)||f\left(\mathbf{\Delta y}|\mathbf{x}_t, \Halt\right)\right]\\
&= \int_{-\infty}^{\infty} \ln \frac{f\left(\mathbf{\Delta y}| \Hnull\right)}{f\left(\mathbf{\Delta y}|\mathbf{x}_t, \Halt\right)} f\left(\mathbf{\Delta y}|\Hnull\right) d{\mathbf{\Delta y}}\\
&= \frac{1}{2}(\mathbf{\Delta v} - \mathbf{\Delta u})^T \mathbf{D}^{-1}(\mathbf{\Delta v} - \mathbf{\Delta u}).
\end{split}
\end{equation}
The optimal value of $\mathbf{x}_t$ for the malicious user in the DRSS-based LVS can be obtained through
\begin{align}\label{optimal_true_DRSS}
\mathbf{x}_t^{\dag} = \argmin_{||\mathbf{x}_t-\mathbf{x}_c||_2 \geq r} \varphi(\mathbf{x}_t).
\end{align}
The likelihood function under $\Halt$ for $\mathbf{x}_t =\mathbf{x}_t^{\dag}$ is given by
\begin{equation}\label{likelihood_DRSS_1_w}
f\left(\mathbf{\Delta y}|\mathbf{x}_t^{\dag}, \Halt\right) = \mathcal{N} (\mathbf{\Delta v}^{\dag}, \mathbf{D}),
\end{equation}
where $\Delta v_m^{\dag} = v_m^{\dag}-v_N^{\dag}$ and $\mathbf{v}^{\dag}$ is obtained by substituting $\mathbf{x}_t^{\dag}$ into $\mathbf{v}$.

\subsection{Performance of the DRSS-based LVS}


In this subsection, we again consider the case where the true location of the malicious user is physically constrained. Specifically, we first analyze the performance of the DRSS-based LVS for an arbitrary $\mathbf{x}_t$, and then present the performance of the DRSS-based LVS for $\mathbf{x}_t = \mathbf{x}_t^{\dag}$ as a special case in this subsection.

Following \eqref{arbitrary}, the specific LRT decision rule of the DRSS-based LVS for any $\mathbf{x}_t$ is given by
\begin{equation}\label{DRSS_ori}
\Lambda\left(\mathbf{\Delta y}\right) \triangleq \frac{f\left(\mathbf{\Delta y}|\mathbf{x}_t, \Halt\right)}
{f\left(\mathbf{\Delta y}|\Hnull\right)}  \begin{array}{c}
\overset{\Hoalt}{\geq} \\
\underset{\Honull}{<}
\end{array}%
\lambda_D,
\end{equation}
where $\Lambda\left(\mathbf{\Delta y}\right)$ is the likelihood ratio of $\mathbf{\Delta y}$ and $\lambda_D$ is a threshold for $\Lambda\left(\mathbf{\Delta y}\right)$.
Substituting \eqref{likelihood_DRSS_0} and \eqref{likelihood_DRSS_1_w} into \eqref{DRSS_ori}, we
obtain $\Lambda\left(\mathbf{\Delta y}\right)$ in $\log$ domain as
\begin{equation*}
\begin{split}
\ln \Lambda\left(\mathbf{\Delta y}\right) &= \frac{1}{2}(\mathbf{\Delta y}-\mathbf{\Delta u})^T \bm{D}^{-1}(\mathbf{\Delta y}-\mathbf{\Delta u})\\
&~~~~- \frac{1}{2}(\mathbf{\Delta y}-\mathbf{\Delta v})^T \bm{D}^{-1}(\mathbf{\Delta y}-\mathbf{\Delta v})\\
&= (\mathbf{\Delta v}-\mathbf{\Delta u})^T \bm{D}^{-1} \mathbf{\Delta y}\\
 &~~~~- \frac{1}{2}(\mathbf{\Delta v}-\mathbf{\Delta u})^T \bm{D}^{-1} (\mathbf{\Delta v}+\mathbf{\Delta u}).
\end{split}
\end{equation*}
Then, we can rewrite the decision rule given in \eqref{DRSS_ori} as
\begin{equation}\label{decision_DRSS}
\mathbb{T}(\mathbf{\Delta y}) \begin{array}{c}
\overset{\Honull}{\geq} \\
\underset{\Hoalt}{<}
\end{array}%
\Gamma_D,
\end{equation}
where $\mathbb{T}(\mathbf{\Delta y})$ is the test statistic given by
\begin{equation}\label{statistic_DRSS}
\mathbb{T}(\mathbf{\Delta y}) \triangleq (\mathbf{\Delta v}-\mathbf{\Delta u})^T \bm{D}^{-1} \mathbf{\Delta y},
\end{equation}
and $\Gamma_D$ is the threshold for $\mathbb{T}(\mathbf{\Delta y})$ given by
\begin{equation}\label{threshold_DRSS}
\Gamma_D \triangleq \ln \lambda_D + \frac{1}{2}(\mathbf{\Delta v}-\mathbf{\Delta u})^T \bm{D}^{-1}(\mathbf{\Delta v}+\mathbf{\Delta u}).
\end{equation}
We then derive the false positive rate, $\alpha_D(\mathbf{x}_t)$, and the detection rate, $\beta_D(\mathbf{x}_t)$, of the DRSS-based LVS for any $\mathbf{x}_t$ in the following theorem.

\begin{theorem}\label{theorem2}
\emph{The false positive and detection rates of the DRSS-based LVS for any $\mathbf{x}_t$ are given by
\begin{align}
\alpha_D(\mathbf{x}_t) &= \mathcal{Q}\left[\frac{\Gamma_D - \left(\mathbf{\Delta v} \!\!-\!\! \mathbf{\Delta u}\right)^T \mathbf{D}^{\!-\!1} \mathbf{\Delta u}}{\sqrt{\left(\mathbf{\Delta v} \!\!-\!\! \mathbf{\Delta u}\right)^T \mathbf{D}^{\!-\!1} \left(\mathbf{\Delta v} \!\!-\!\! \mathbf{\Delta u}\right)}}\right]\notag \\
&= \mathcal{Q}\left[\frac{\ln \lambda_D + \frac{1}{2}\left(\mathbf{\Delta v} \!\!-\!\! \mathbf{\Delta u}\right)^T \mathbf{D}^{\!-\!1} \left(\mathbf{\Delta v}\!\!-\!\!\mathbf{\Delta u}\right)}{\sqrt{\left(\mathbf{\Delta v} \!\!-\!\! \mathbf{\Delta u}\right)^T \mathbf{D}^{\!-\!1} \left(\mathbf{\Delta v} \!\!-\!\! \mathbf{\Delta u}\right)}}\right],\label{alpha_D_xt}\\
\beta_D(\mathbf{x}_t) &= \mathcal{Q}\left[\frac{\Gamma_D - \left(\mathbf{\Delta v} \!\!-\!\! \mathbf{\Delta u}\right)^T \mathbf{D}^{\!-\!1} \mathbf{\Delta v}}{\sqrt{\left(\mathbf{\Delta v} \!\!-\!\! \mathbf{\Delta u}\right)^T \mathbf{D}^{\!-\!1} \left(\mathbf{\Delta v} \!\!-\!\! \mathbf{\Delta u}\right)}}\right]\notag\\
&= \mathcal{Q}\left[\frac{\ln \lambda_D - \frac{1}{2}\left(\mathbf{\Delta v} \!\!-\!\! \mathbf{\Delta u}\right)^T \mathbf{D}^{\!-\!1} \left(\mathbf{\Delta v}\!\!-\!\!\mathbf{\Delta u}\right)}{\sqrt{\left(\mathbf{\Delta v} \!\!-\!\! \mathbf{\Delta u}\right)^T \mathbf{D}^{\!-\!1} \left(\mathbf{\Delta v} \!\!-\!\! \mathbf{\Delta u}\right)}}\right].\label{beta_D_xt}
\end{align}}
\end{theorem}
\begin{proof}
Using \eqref{likelihood_DRSS_0}, \eqref{likelihood_DRSS_1_w}, and \eqref{statistic_DRSS}, the distributions of $\mathbb{T}(\mathbf{\Delta y})$ under $\Hnull$ and $\Halt$ are derived as follows
\begin{small}
\begin{align}
&\mathbb{T}(\mathbf{\Delta y}) |\Hnull \notag \\
&\!\sim\! \mathcal{N}\left(\left(\mathbf{\Delta v} \!\!-\!\! \mathbf{\Delta u}\right)^T \mathbf{D}^{\!-\!1} \mathbf{\Delta u}, \left(\mathbf{\Delta v} \!\!-\!\! \mathbf{\Delta u}\right)^T \mathbf{D}^{\!-\!1} \left(\mathbf{\Delta v} \!\!-\!\! \mathbf{\Delta u}\right)\right),\label{D_0}\\
&\mathbb{T}(\mathbf{\Delta y}) |\Halt \notag\\
&\!\sim\! \mathcal{N}\left(\left(\mathbf{\Delta v} \!\!-\!\! \mathbf{\Delta u}\right)^{\!T\!} \mathbf{D}^{\!-\!1} \mathbf{\Delta v}, \left(\mathbf{\Delta v} \!\!-\!\! \mathbf{\Delta u}\right)^T \mathbf{D}^{\!-\!1} \left(\mathbf{\Delta v} \!\!-\!\! \mathbf{\Delta u}\right)\right).\label{D_1}
\end{align}
\end{small}
As per the decision rule in \eqref{decision_DRSS}, the false positive and detection rates are given by
\begin{align}
\alpha_D(\mathbf{x}_t) &\triangleq \Pr\left(\mathbb{T}(\mathbf{\Delta y}) \geq \Gamma_D |\Hnull \right),\label{false_D}\\
\beta_D(\mathbf{x}_t) &\triangleq \Pr\left(\mathbb{T}(\mathbf{\Delta y}) \geq \Gamma_D |\Halt \right).\label{detection_D}
\end{align}
Substituting \eqref{D_0} and \eqref{D_1} into \eqref{false_D} and \eqref{detection_D}, respectively, we obtain the results in \eqref{alpha_D} and \eqref{beta_D} after some algebraic manipulations.
\end{proof}

For $\mathbf{x}_t = \mathbf{x}_t^{\dag}$, the LRT decision rule of the DRSS-based LVS is given by
\begin{equation}\label{DRSS_ori}
\Lambda^{\ast}\left(\mathbf{\Delta y}\right) \triangleq \frac{f\left(\mathbf{\Delta y}|\mathbf{x}_t, \Halt\right)}
{f\left(\mathbf{\Delta y}|\Hnull\right)}  \begin{array}{c}
\overset{\Hoalt}{\geq} \\
\underset{\Honull}{<}
\end{array}%
\lambda_D^{\ast},
\end{equation}
where $\Lambda^{\ast}\left(\mathbf{\Delta y}\right)$ is the likelihood ratio of $\mathbf{\Delta y}$ for $\mathbf{x}_t = \mathbf{x}_t^{\dag}$ and $\lambda_D^{\ast}$ is a threshold for $\Lambda^{\ast}\left(\mathbf{\Delta y}\right)$.
Following Theorem~\ref{theorem2}, the false positive and detection rates of the DRSS-based LVS for $\mathbf{x}_t = \mathbf{x}_t^{\dag}$ are given by
\begin{align}
\alpha_D^{\ast} &= \mathcal{Q}\left[\frac{\ln \lambda_D^{\ast} + \frac{1}{2}\left(\mathbf{\Delta v}^{\dag} \!\!-\!\! \mathbf{\Delta u}\right)^T \mathbf{D}^{\!-\!1} \left(\mathbf{\Delta v}^{\dag}\!\!-\!\!\mathbf{\Delta u}\right)}{\sqrt{\left(\mathbf{\Delta v}^{\dag} \!\!-\!\! \mathbf{\Delta u}\right)^T \mathbf{D}^{\!-\!1} \left(\mathbf{\Delta v}^{\dag} \!\!-\!\! \mathbf{\Delta u}\right)}}\right],\label{alpha_D} \\
\beta_D^{\ast} &= \mathcal{Q}\left[\frac{\ln \lambda_D^{\ast} - \frac{1}{2}\left(\mathbf{\Delta v}^{\dag} \!\!-\!\! \mathbf{\Delta u}\right)^T \mathbf{D}^{\!-\!1} \left(\mathbf{\Delta v}^{\dag}\!\!-\!\!\mathbf{\Delta u}\right)}{\sqrt{\left(\mathbf{\Delta v}^{\dag} \!\!-\!\! \mathbf{\Delta u}\right)^T \mathbf{D}^{\!-\!1} \left(\mathbf{\Delta v}^{\dag} \!\!-\!\! \mathbf{\Delta u}\right)}}\right].\label{beta_D}
\end{align}

Again, note that the results provided in \eqref{alpha_D_xt} and \eqref{beta_D_xt} are for any $\mathbf{x}_t$, which are more general than that provided in \eqref{alpha_D} and \eqref{beta_D}. That is, $\alpha_D^{\ast} = \alpha_D(\mathbf{x}_t^{\dag})$ and $\beta_D^{\ast} = \beta_D(\mathbf{x}_t^{\dag})$. By using \eqref{alpha_D_xt} and \eqref{beta_D_xt}, we can compare the performance of the DRSS-based LVS with that of the RSS-based LVS in a general scenario.

\subsection{Comparison between the RSS-based LVS and the DRSS-based LVS}

We now present the following theorem with regard to the comparison between the RSS-based LVS and the DRSS-based LVS.

\begin{theorem}\label{theorem3}
\emph{For any $\mathbf{x}_t$, we have $\alpha_R^o(\mathbf{x}_t) = \alpha_D(\mathbf{x}_t)$ and $\beta_R^o(\mathbf{x}_t) = \beta_D(\mathbf{x}_t)$ for $\lambda_R = \lambda_D$. That is, for any $\mathbf{x}_t$ the performance of the RSS-based LVS with $p_x=p_x^o(\mathbf{x}_t)$ is identical to the performance of the DRSS-based LVS.}
\end{theorem}
\begin{proof}
Based on \eqref{alpha_R_xt}, \eqref{beta_R_xt}, \eqref{alpha_D_xt}, and \eqref{beta_D_xt}, we can see that $\alpha_R^o(\mathbf{x}_t)$, $\beta_R^o(\mathbf{x}_t)$, $\alpha_D(\mathbf{x}_t)$, and $\beta_D(\mathbf{x}_t)$ are all in the form of a $\mathcal{Q}$ function. We denote $\alpha_R^o(\mathbf{x}_t) = \mathcal{Q}(\zeta_R^o)$, $\beta_R^o(\mathbf{x}_t) = \mathcal{Q}(\eta_R^o)$, $\alpha_D(\mathbf{x}_t) = \mathcal{Q}(\zeta_D)$, and $\beta_D(\mathbf{x}_t) = \mathcal{Q}(\eta_D)$. In order to prove $\alpha_R^o(\mathbf{x}_t) = \alpha_D(\mathbf{x}_t)$ and $\beta_R^o(\mathbf{x}_t) = \beta_D(\mathbf{x}_t)$ for $\lambda_R = \lambda_D$, we only need to prove $\zeta_R^o - \eta_R^o = \zeta_D - \eta_D$. As per \eqref{alpha_R_xt}, \eqref{beta_R_xt}, \eqref{alpha_D_xt}, and \eqref{beta_D_xt}, in order to prove $\zeta_R^o - \eta_R^o = \zeta_D - \eta_D$ (such as to prove Theorem~\ref{theorem3}) we have to prove the following equation
\begin{align}\label{equation_prove1}
\left(\mathbf{w} \!\!-\!\! \mathbf{u}\right)^T \mathbf{R}^{\!-\!1} \left(\mathbf{w}\!\!-\!\!\mathbf{u}\right) =
\left(\mathbf{\Delta v} \!\!-\!\! \mathbf{\Delta u}\right)^T \mathbf{D}^{\!-\!1} \left(\mathbf{\Delta v}\!\!-\!\!\mathbf{\Delta u}\right).
\end{align}
Based on the singular value decomposition (SVD) of $\mathbf{R}$, we can transform the RSS observation vector $\mathbf{y}$ into another observation vector $\mathbf{y}'$ by rotating and scaling\footnote{The covariance matrix $\mathbf{R}$ is a real positive-definite symmetric matrix, and thus the SVD of $\mathbf{R}$ can be written as $\mathbf{R} = \mathbf{S}\mathbf{R}'\mathbf{S}^T$. As such, $\mathbf{y}'$ is given by $\mathbf{y}' = \mathbf{R}'^{\frac{1}{2}}\mathbf{S}\mathbf{y}$ and the covariance matrix of $\mathbf{y}'$ will be $\mathbf{I}_N$.}. We can then obtain the DRSS observations from $\mathbf{y}'$ instead of $\mathbf{y}$. The transformation from $\mathbf{y}$ to $\mathbf{y}'$ is unique since the singular values of $\mathbf{R}$ are unique. In addition, $\mathbf{y}$ follows a multivariate normal distribution. As such, the transformation from $\mathbf{y}$ to $\mathbf{y}'$ keeps all the properties of $\mathbf{y}$ in $\mathbf{y}'$, which means the performance of an LVS based on $\mathbf{y}$ is identical to the performance of an LVS based on $\mathbf{y}'$ \cite{scharf1994matched,kay2003an}. Therefore, in order to prove Theorem~3 we only have to prove \eqref{equation_prove1} for $\mathbf{R}=\mathbf{I}_N$.
Denoting $\mathbf{g} = \mathbf{v} - \mathbf{u}$, we have $\Delta v_m- \Delta u_m = g_m - g_N$.
Substituting $\mathbf{R} = \mathbf{I}_N$ into $\mathbf{w}$ given in \eqref{definition_w}, we obtain
\begin{align}
\mathbf{w}-\mathbf{u}= \mathbf{g} - \frac{\mathbf{g}^T \mathbf{R}^{-1}\mathbf{1}_N}
{\mathbf{1}_N^T\mathbf{R}^{-1}\mathbf{1}_N}\mathbf{1}_N
= \mathbf{g} - \left(\frac{1}{N}\sum_{j=1}^N g_j\right)\mathbf{1}_N. \notag
\end{align}
With regard to the left side of \eqref{equation_prove1}, for $\mathbf{R}=\mathbf{I}_N$ we have
\begin{align}\label{left_final}
&\left(\mathbf{w} \!\!-\!\! \mathbf{u}\right)^T \mathbf{R}^{\!-\!1} \left(\mathbf{w}\!\!-\!\!\mathbf{u}\right) = \sum_{i=1}^N \left(g_i - \frac{1}{N}\sum_{j=1}^N g_j \right)^2\notag\\
&= \sum_{i=1}^N \left[g_i^2 - \frac{2}{N}g_i\sum_{j=1}^N g_j + \frac{1}{N^2}\left(\sum_{j=1}^N g_j\right)^2 \right]\notag\\
&= \left[\sum_{i=1}^N g_i^2 \!-\! \frac{2}{N}\left(\sum_{i=1}^N g_i\right)\left(\sum_{j=1}^N g_j\right) \!+\! \frac{1}{N}\left(\sum_{j=1}^N g_j\right)^2 \right]\notag \\
&=  \left[\sum_{i=1}^N g_i^2 - \frac{1}{N}\left(\sum_{i=1}^{N} g_i\right)^2\right].
\end{align}
As per the definition of $\mathbf{D}$ given in \eqref{D_definition}, for $\mathbf{R}=\mathbf{I}_N$ we have
\begin{align}
\mathbf{D} = \mathbf{I}_{N-1} + \mathbf{1}_{(N-1) \times (N-1)},
\end{align}
where $\mathbf{1}_{(N-1) \times (N-1)})$ is the $(N-1) \times (N-1)$ matrix with all elements set to unity. Then, based on the Sherman-Morrison formula \cite{sherman1950adjustment}, we have
\begin{align}\label{D_inverse}
\mathbf{D}^{-1} &= \left[\mathbf{I}_{N-1} + \mathbf{1}_{(N-1) \times (N-1)}\right]^{-1}\notag\\
&=\left[\mathbf{I}_{N-1} + \mathbf{1}_{(N-1)} \times \mathbf{1}_{(N-1)}^T \right]^{-1}\notag\\
&=\left[\mathbf{I}_{N-1}^{-1} - \frac{\mathbf{I}_{N-1}^{-1}\mathbf{1}_{(N-1) \times (N-1)}\mathbf{I}_{N-1}^{-1}}{1+\mathbf{1}_{(N-1)}^T \mathbf{I}_{N-1}^{-1}\mathbf{1}_{(N-1)}} \right]\notag\\
&=\left[\mathbf{I}_{N-1} - \frac{\mathbf{1}_{(N-1) \times (N-1)}}{N}\right].
\end{align}
Substituting \eqref{D_inverse} into the right side of \eqref{equation_prove1}, we have
\begin{align}\label{right_final}
&\left(\mathbf{\Delta v} \!\!-\!\! \mathbf{\Delta u}\right)^T \mathbf{D}^{\!-\!1} \left(\mathbf{\Delta v}\!\!-\!\!\mathbf{\Delta u}\right) \notag\\
&\!=\! \left(\mathbf{\Delta v} \!\!-\!\! \mathbf{\Delta u}\right)^T \left[\mathbf{I}_{N\!-\!1} \!\!-\!\! \frac{\mathbf{1}_{(N\!-\!1) \times (N\!-\!1)}}{N}\right] \left(\mathbf{\Delta v}\!\!-\!\!\mathbf{\Delta u}\right)\notag\\
&\!=\! \left(\mathbf{\Delta v} \!\!-\!\! \mathbf{\Delta u}\right)^T \mathbf{I}_{N\!-\!1} \left(\mathbf{\Delta v}\!\!-\!\!\mathbf{\Delta u}\right)\notag\\
&~~~~\!\!-\!\!\frac{1}{N}\left(\mathbf{\Delta v} \!\!-\!\! \mathbf{\Delta u}\right)^T \mathbf{1}_{(N\!-\!1)} \times \mathbf{1}_{(N\!-\!1)}^T \left(\mathbf{\Delta v}\!\!-\!\!\mathbf{\Delta u}\right)\notag\\
&\!=\! \sum_{i=1}^{N-1}\left(g_i-g_N\right)^2
-\frac{1}{N} \left[\sum_{i=1}^{N-1}\left(g_i-g_N\right)\right]^2 \notag \\
&\!=\! \sum_{i=1}^{N}\left(g_i-g_N\right)^2
-\frac{1}{N} \left[\sum_{i=1}^{N}\left(g_i-g_N\right)\right]^2 \notag \\
&\!=\! \sum_{i=1}^{N}\left(g_i-g_N\right)^2
\!-\!\frac{1}{N} \sum_{i=1}^{N}\left(g_i-g_N\right)\left[\sum_{j=1}^N\left(g_j-g_N\right)\right] \notag \\
&\!=\! \left[\sum_{i=1}^N g_i^2 - \frac{1}{N}\left(\sum_{i=1}^{N} g_i\right)^2\right].
\end{align}
Comparing \eqref{left_final} with \eqref{right_final}, we can see that we have proved \eqref{equation_prove1} for $\mathbf{R}=\mathbf{I}_N$. This completes the proof of Theorem~\ref{theorem3}.
\end{proof}

We note that the result provided in Theorem~\ref{theorem3} is valid for any $\mathbf{R}$, \emph{i.e.}, for any kind of shadowing (correlated or uncorrelated). We also note that in Theorem~\ref{theorem3} the condition to guarantee the RSS-based LVS being identical to the DRSS-based LVS is that $p_x = p_x^o(\mathbf{x}_t)$. This condition forces the malicious user to optimize his transmit power based on the given $\mathbf{x}_t$ in the RSS-based LVS, but not in the DRSS-based LVS. Without this condition, the comparison result between the RSS-based LVS and the DRSS-based LVS is present in the following corollary.

\begin{corollary}\label{corollary1}
\emph{For any $\mathbf{x}_t$, the performance of the RSS-based LVS with $p_x \neq p_x^o(\mathbf{x}_t)$ is better than the performance of the DRSS-based LVS.}
\end{corollary}

\begin{proof}
For any $p_x$ and $\mathbf{x}_t$, the LRT decision rule of the RSS-based LVS is given by
\begin{equation}\label{RSS_known_gen}
\Lambda\left(\mathbf{y}\right) \triangleq \frac{f\left(\mathbf{y}|p_x, \mathbf{x}_t, \Halt\right)}{f\left(\mathbf{y}|\Hnull\right)}  \begin{array}{c}
\overset{\Hoalt}{\geq} \\
\underset{\Honull}{<}
\end{array}%
\lambda_R,
\end{equation}
where $\Lambda\left(\mathbf{y}\right)$ is the likelihood ratio of $\mathbf{y}$ and $\lambda_R$ is a threshold for $\Lambda\left(\mathbf{y}\right)$. Following Theorem~\ref{theorem1}, the false positive and detection rates of the RSS-based LVS for any $p_x$ and $\mathbf{x}_t$ are given by
\begin{align}
\alpha_R(p_x, \mathbf{x}_t) &= \mathcal{Q}\left[\frac{\ln \lambda_R + \frac{1}{2}\left(\mathbf{v} \!\!-\!\! \mathbf{u}\right)^T \mathbf{R}^{\!-\!1} \left(\mathbf{v}\!\!-\!\!\mathbf{u}\right)}{\sqrt{\left(\mathbf{v} \!\!-\!\! \mathbf{u}\right)^T \mathbf{R}^{\!-\!1} \left(\mathbf{v} \!\!-\!\! \mathbf{u}\right)}}\right],\label{alpha_R_gen} \\
\beta_R(p_x, \mathbf{x}_t) &= \mathcal{Q}\left[\frac{\ln \lambda_R - \frac{1}{2}\left(\mathbf{v} \!\!-\!\! \mathbf{u}\right)^T \mathbf{R}^{\!-\!1} \left(\mathbf{v}\!\!-\!\!\mathbf{u}\right)}{\sqrt{\left(\mathbf{v} \!\!-\!\! \mathbf{u}\right)^T \mathbf{R}^{\!-\!1} \left(\mathbf{v} \!\!-\!\! \mathbf{u}\right)}}\right].\label{beta_R_gen}
\end{align}
Then, Corollary~1 can be presented in math as that given $p_x \neq p_x^o(\mathbf{x}_t)$, we have $\beta_R(p_x, \mathbf{x}_t) > \beta_D(\mathbf{x}_t)$ for $\alpha_R(p_x, \mathbf{x}_t)= \alpha_D(\mathbf{x}_t)$ or $\alpha_R(p_x, \mathbf{x}_t)< \alpha_D(\mathbf{x}_t)$ for $\beta_R(p_x, \mathbf{x}_t) = \beta_D(\mathbf{x}_t)$.
Given the proof of Theorem~\ref{theorem3}, in order to prove Corollary~\ref{corollary1} we only have to prove the following equation
\begin{align}\label{equation_prove2}
\left(\mathbf{v} \!\!-\!\! \mathbf{u}\right)^T \mathbf{R}^{\!-\!1} \left(\mathbf{v}\!\!-\!\!\mathbf{u}\right) >
\left(\mathbf{\Delta v} \!\!-\!\! \mathbf{\Delta u}\right)^T \mathbf{D}^{\!-\!1} \left(\mathbf{\Delta v}\!\!-\!\!\mathbf{\Delta u}\right).
\end{align}
Following similar manipulations in \eqref{left_final}, for $\mathbf{R} = \mathbf{I}_N$ we have
\begin{align}\label{general_left}
\left(\mathbf{v} \!\!-\!\! \mathbf{u}\right)^T \mathbf{R}^{\!-\!1} \left(\mathbf{v}\!\!-\!\!\mathbf{u}\right) = \sum_{i=1}^N g_i^2.
\end{align}
Since the malicious user's true location cannot be the same as his claimed location, \emph{i.e.}, $\mathbf{x}_t \neq \mathbf{x}_c$, we have $\mathbf{v} \neq \mathbf{u}$ and $\left(\sum_{i=1}^{N} g_i\right)^2 > 0$. As such, as per \eqref{left_final} and \eqref{general_left} we have
\begin{align}\label{general_large}
\left(\mathbf{v} \!\!-\!\! \mathbf{u}\right)^T \mathbf{R}^{\!-\!1} \left(\mathbf{v}\!\!-\!\!\mathbf{u}\right) > \left(\mathbf{w} \!\!-\!\! \mathbf{u}\right)^T \mathbf{R}^{\!-\!1} \left(\mathbf{w}\!\!-\!\!\mathbf{u}\right).
\end{align}
Based on \eqref{equation_prove1} and \eqref{general_large}, we have proved \eqref{equation_prove2}, which completes the proof of Corollary~\ref{corollary1}.
\end{proof}

We note that Corollary~\ref{corollary1} presents a fair comparison between the RSS-based LVS and the DRSS-based LVS when the malicious user does not know the transmit power of the legitimate user and thus cannot optimize his transmit power.


Under the best attack strategies of the malicious user, the comparison result between the RSS-based LVS and the DRSS-based LVS is present in the following corollary.

\begin{corollary}\label{corollary2}
\emph{We have $\alpha_R^{\ast}=\alpha_D^{\ast}$ and $\beta_R^{\ast}=\beta_D^{\ast}$ for $\lambda_R^{\ast} = \lambda_D^{\ast}$. That is, the performance of the RSS-based LVS for $p_x = p_x^{\ast}$ and $\mathbf{x}_t = \mathbf{x}_t^{\ast}$ is identical to the performance of the DRSS-based LVS for $\mathbf{x}_t = \mathbf{x}_t^{\dag}$.}
\end{corollary}
\begin{proof}
Based on Theorem~\ref{theorem3}, in order to prove Corollary~\ref{corollary2} we only have to prove $\mathbf{x}_t^{\ast} = \mathbf{x}_t^{\dag}$. We note that $\mathbf{x}_t^{\ast}$ and  $\mathbf{x}_t^{\dag}$ are obtained through minimizing $\phi(p_x^o(\mathbf{x}_t), \mathbf{x}_t)$ and $\varphi(\mathbf{x}_t)$, respectively. As such, in order to prove $\mathbf{x}_t^{\ast} = \mathbf{x}_t^{\dag}$, it suffices to prove $\phi(p_x^o(\mathbf{x}_t), \mathbf{x}_t) =\varphi(\mathbf{x}_t)$. As per \eqref{KL_divergence_min} and \eqref{KL_divergence_DRSS}, we can see that we have proved $\phi(p_x^o(\mathbf{x}_t), \mathbf{x}_t) =\varphi(\mathbf{x}_t)$ in \eqref{equation_prove1}.
\end{proof}

We note that Corollary~\ref{corollary2} presents a comparison between the performance limits of the RSS-based LVS and the DRSS-based LVS. In the proof of Corollary~\ref{corollary2}, we also prove that the malicious user's optimal true locations for the RSS-based LVS and the DRSS-based LVS are the same. We also note that the analysis and results reported in this work are not directly applicable to the colluding threat scenario (where multiple colluding adversaries attack the LVS). Future studies may wish to explore these more sophisticated attacks,  in the context of correlated fading channels. However, although such sophisticated attacks will obviously lead to poorer LVS performance, a conjecture is that the trends discovered here with regard to the impact of correlated shadowing on LVS performance will persist.

\section{Numerical Results}\label{sec_numerical}

We now present numerical results to verify the accuracy of our provided analysis. We also provide some insights on the impact of the spatially correlated shadowing on the performance of the RSS-based LVS and the DRSS-based LVS.


\begin{figure}[!t]
    \begin{center}
   {\includegraphics[width=3.5in, height=3.0in]{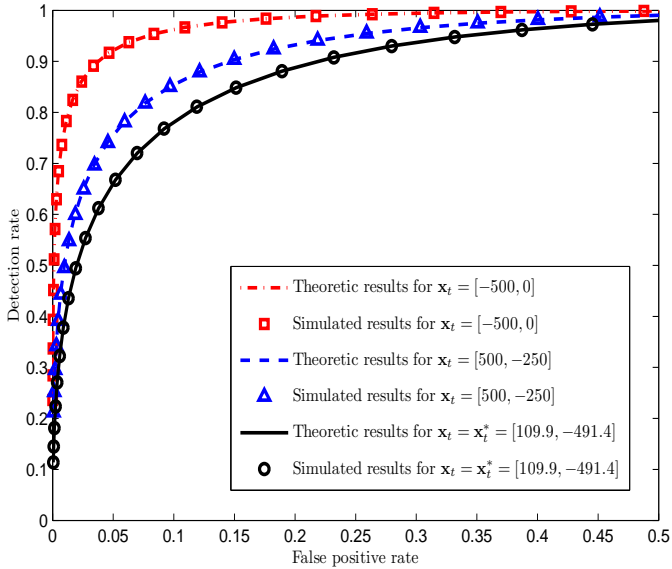}}
    \caption{ROC curves of the RSS-based LVS for $\sigma_{dB} = 7.5$, $D_c = 50$m, $r = 500$m, $p_x = p_x^o(\mathbf{x}_t)$, and $N = 3$ \big($\mathbf{x}_1 = [-250, 10]$, $\mathbf{x}_2 = [0, -10]$, and $\mathbf{x}_3 = [250, 10]$\big).}\label{fig:theorem1}
    \end{center}
\end{figure}

Although we have simulated a wide range of system settings, the associated settings for the results shown in this work (unless otherwise stated) are as follows. In the simulations specifically shown here, the BSs and the claimed locations are deployed in a rectangular area 500m by 20m. The origin is set at the center of the rectangular area, with the x-coordinate taken along the length, and the y-coordinate taken along the width. The claimed locations of both legitimate and malicious users are set such as $\mathbf{x}_c = [50,5]$, which is also the true location of the legitimate user. The locations of all BSs are provided in the caption of each figure, and all BSs collect measurements from the legitimate and malicious users. The path loss exponent is set to $\gamma=3$, and the reference power is set to $p=-10$~dB at $d=1$m.

\begin{figure}[!t]
    \begin{center}
   {\includegraphics[width=3.5in, height=3.0in]{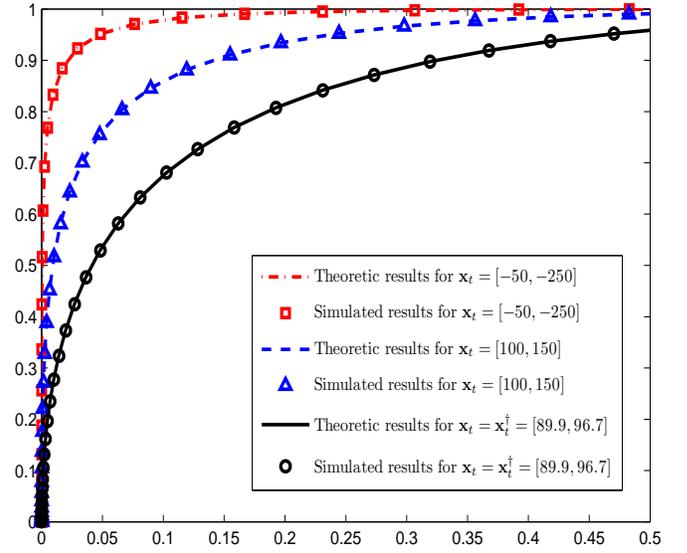}}
    \caption{ROC curves of the DRSS-based LVS for $\sigma_{dB} = 5$, $D_c = 50$m, $r = 100$m, and $N = 4$ \big($\mathbf{x}_1 = [201.4, -9.0]$, $\mathbf{x}_2 = [-161.7, 9.3]$, $\mathbf{x}_3 = [-97.4, 1.2]$, and $\mathbf{x}_4 = [91.5, 2.4]$\big).}\label{fig:theorem2}
    \end{center}
\end{figure}


In Fig.~\ref{fig:theorem1}, we present the Receiver Operating Characteristic (ROC) curves of the RSS-based LVS. In order to obtain this figure, we have set the BSs at regular intervals (250m) on each side of the rectangular area. In this figure, we first observe that the Monte Carlo simulations precisely match the theoretic results, confirming our analysis in Theorem~1. We also observe that the ROC curves for $\mathbf{x}_t \neq \mathbf{x}_t^{\ast}$ dominate the ROC curve for $\mathbf{x}_t = \mathbf{x}_t^{\ast}$. This observation indicates that if the malicious user does not optimize his true location, it will be easier for the RSS-based LVS to detect the malicious user.  In summary, the ROC curve for $\mathbf{x}_t = \mathbf{x}_t^{\ast}$ (analysis presented in \eqref{alpha_R} and \eqref{beta_R}) provides a lower bound for the performance of the RSS-based LVS.

In Fig.~\ref{fig:theorem2}, we present the ROC curves of the DRSS-based LVS. In order to obtain this figure, we have deployed the BSs randomly inside the rectangular area, which relates to a scenario where authorized vehicles represent the BSs. In this scenario the authorized vehicles already have their locations authenticated, and they are used as anchor points in authenticating the positions of yet-to-be authorized vehicles. In this figure, we first observe that the Monte Carlo simulations precisely match the theoretic results, confirming our analysis
in Theorem~2. We also observe that the ROC curves for $\mathbf{x}_t \neq \mathbf{x}_t^{\dag}$ dominate the ROC curve for $\mathbf{x}_t = \mathbf{x}_t^{\dag}$. Again, this observation demonstrates the importance of optimally choosing the true location for the malicious user. To conclude, the ROC curve for $\mathbf{x}_t = \mathbf{x}_t^{\dag}$ (analysis presented in \eqref{alpha_D} and \eqref{beta_D}) provides a lower bound for the performance of the DRSS-based LVS.

In Fig.~\ref{fig:theorem3}, we present the ROC curves of the RSS-based LVS and the DRSS-based LVS. In order to obtain this figure, we have set one of the BSs at one side of the rectangular area and deployed the other two BSs randomly inside the rectangular area. This mimics the scenario in which only one fixed BS is available and we have to conduct location verification with the help of two already-authorized vehicles. In this figure, we first observe that the RSS-based LVS for $p_x = p_x^o(\mathbf{x}_t)$ and the DRSS-based LVS achieve identical performance (identical ROC curves). This demonstrates that as long as the malicious user optimizes his transmit power (as per his true location) the RSS-based LVS is identical to the DRSS-based LVS, which confirms the analytical comparison between the RSS-based LVS and the DRSS-based LVS presented in Theorem~3. We also observe that the ROC curves of the RSS-based LVS for $p_x \neq p_x^o(\mathbf{x}_t)$ dominate the ROC curves of the DRSS-based LVS. This observation confirms that if the malicious user does not optimize his transmit power, the RSS-based LVS achieves a better performance than the DRSS-based LVS, which is provided in Corollary~\ref{corollary1}. This indicates that the RSS-based LVS is subjectively better than the DRSS-based LVS since the performance of the DRSS-based LVS is independent of the malicious user's transmit power and the determination of the optimal transmit power for the malicious user is no longer
required in the DRSS-based LVS. In the simulations of Fig.~\ref{fig:theorem3}, we confirmed that the malicious user's optimal true location for the RSS-based LVS is the same as that for the DRSS-based LVS, \emph{i.e.}, $\mathbf{x}_t^{\ast} = \mathbf{x}_t^{\dag}$. As such, Fig.~\ref{fig:theorem3} also confirms our analysis provided in Corollary~\ref{corollary2}.

\begin{figure}[!t]
    \begin{center}
   {\includegraphics[width=3.5in, height=3.0in]{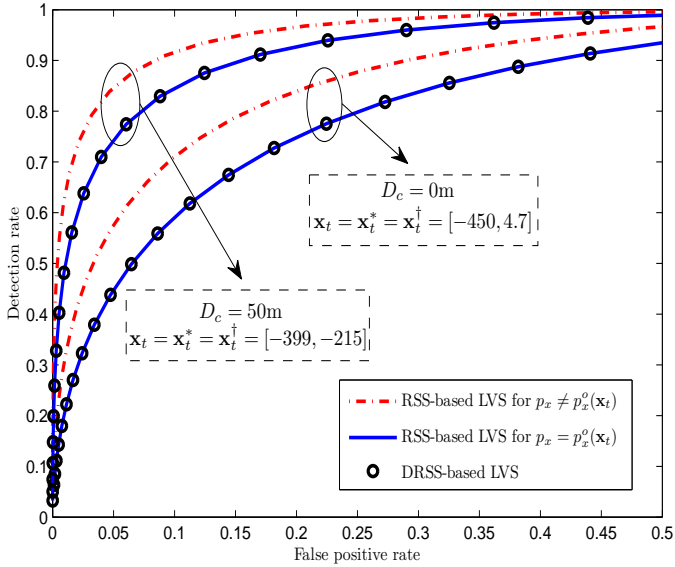}}
    \caption{ROC curves of the RSS-based LVS and the DRSS-based LVS for $\sigma_{dB} = 5$, $D_c = 50$m, $r = 100$m, and $N = 3$ \big($\mathbf{x}_1 = [0, 10]$, $\mathbf{x}_2 = [131.4, -9.3]$, and $\mathbf{x}_3 = [20.6, -0.9]$\big).}\label{fig:theorem3}
    \end{center}
\end{figure}

\begin{figure}[!t]
    \begin{center}
   {\includegraphics[width=3.5in, height=3.0in]{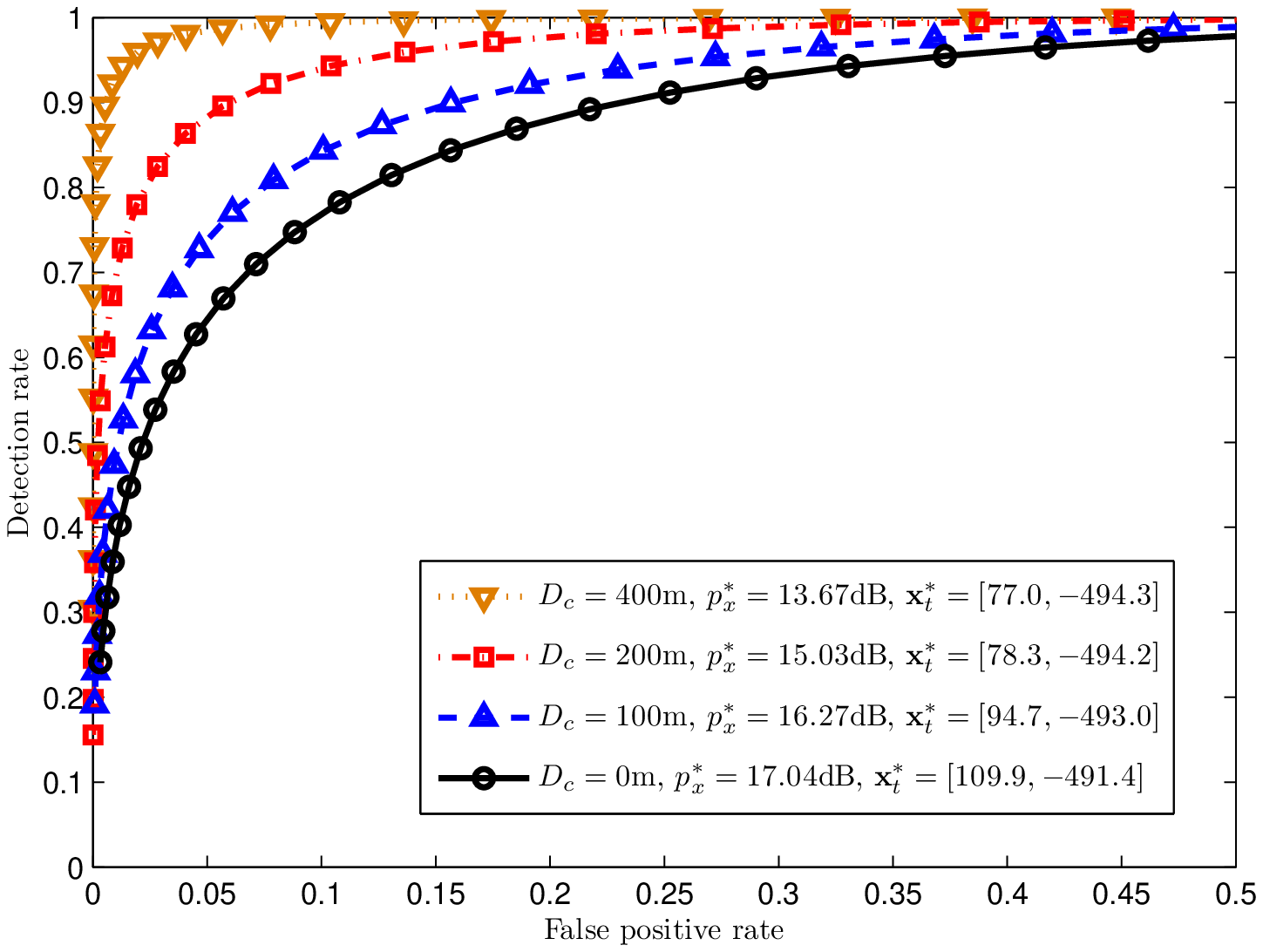}}
    \caption{ROC curves of the RSS-based LVS for $\sigma_{dB} = 7.5$, $r = 500$m, $p_x = p_x^{\ast}$, $\mathbf{x}_t = \mathbf{x}_t^{\ast}$, and $N = 3$ \big($\mathbf{x}_1 = [-250, 10]$, $\mathbf{x}_2 = [0, -10]$, and $\mathbf{x}_3 = [250, 10]$\big).}\label{fig:correlation_RSS}
    \end{center}
\end{figure}

\begin{figure}[!t]
    \begin{center}
   {\includegraphics[width=3.5in, height=3.0in]{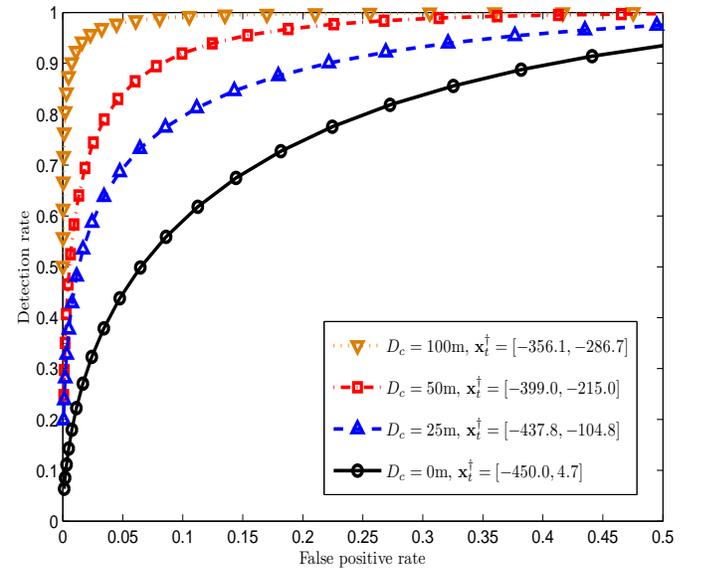}}
    \caption{ROC curves of the DRSS-based LVS for $\sigma_{dB} = 5$, $D_c = 50$m, $r = 100$m, and $N = 3$ \big($\mathbf{x}_1 = [0, 10]$, $\mathbf{x}_2 = [131.4, -9.3]$, and $\mathbf{x}_3 = [20.6, -0.9]$\big).}\label{fig:correlation_DRSS}
    \end{center}
\end{figure}

In Fig.~\ref{fig:correlation_RSS} and Fig.~\ref{fig:correlation_DRSS}, we investigate the impact of the spatial correlation of the shadowing on the performance of the RSS-based LVS and the DRSS-based LVS, where $D_c = 0$m corresponds to the case with uncorrelated shadowing. In Fig.~\ref{fig:correlation_RSS}, we set $p_x = p_x^{\ast}$ and $\mathbf{x}_t = \mathbf{x}_t^{\ast}$ for the RSS-based LVS. From \eqref{optimal_px} and \eqref{optimal_true}, we can see that both $p_x^{\ast}$ and $\mathbf{x}_t^{\ast}$ are dependent on the spatial correlation of the shadowing (they are both functions of $D_c$), and the exact values of $p_x^{\ast}$ and $\mathbf{x}_t^{\ast}$ corresponding to each $D_c$ are also provided in Fig.~\ref{fig:correlation_RSS}.
In this figure, we first observe the ROC curve moves toward the upper left corner (\emph{i.e.}, the area under the ROC curve increases)
as $D_c$ increases, which shows that the performance of the RSS-based LVS becomes better as $D_c$ increases. This observation demonstrates that the spatial correlation of the shadowing improves the detection performance of the RSS-based LVS. We note that the above performance improvement due to the spatial correlation of the shadowing is only achieved under the condition  $p_x = p_x^{\ast}$ and $\mathbf{x}_t = \mathbf{x}_t^{\ast}$. If the malicious user is physically limited at some specific location $\mathbf{x}_t$ and he optimizes his transmit power as per $\mathbf{x}_t$, \emph{i.e.}, $p_x = p_x^o(\mathbf{x}_t)$, the spatial correlation of the shadowing does not have a monotonic impact on the performance of the RSS-based LVS. As per Theorem~3 and Corollary~\ref{corollary2}, the ROC curves provided in Fig.~\ref{fig:correlation_RSS} are also valid for the DRSS-based LVS, in which we have to set $\mathbf{x}_t = \mathbf{x}_t^{\dag}$. As such, we can conclude that the spatial correlation of the shadowing also improves the detection performance of the DRSS-based LVS. Also, for a determined $\mathbf{x}_t$ the spatial correlation does not have a monotonic impact on the performance of the DRSS-based LVS. For confirmation, we also provide the ROC curves for the DRSS-based LVS in Fig.~\ref{fig:correlation_DRSS} under different settings. The same conclusion on the impact of spatial correlation of the shadowing can be drawn from Fig.~\ref{fig:correlation_DRSS}.

\begin{figure}[!t]
    \begin{center}
   {\includegraphics[width=3.5in, height=3.0in]{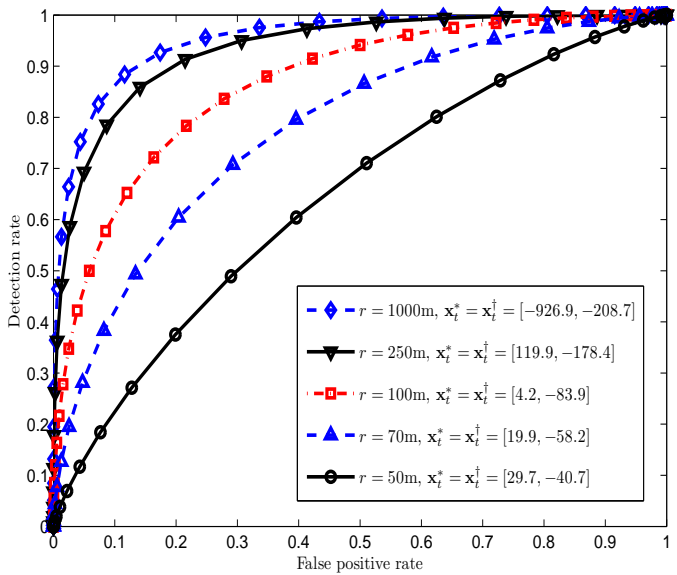}}
    \caption{ROC curves of the RSS-based LVS ($p_x = p_x^{\ast}$) and the DRSS-based LVS for $\sigma_{dB} = 5$, $D_c = 50$m, $\mathbf{x}_t = \mathbf{x}_t^{\ast} = \mathbf{x}_t^{\dag}$, and $N = 3$ \big($\mathbf{x}_1 = [0, 10]$, $\mathbf{x}_2 = [131.4, -9.3]$, and $\mathbf{x}_3 = [20.6, -0.9]$\big).}\label{fig:threshold}
    \end{center}
\end{figure}

In Fig.~\ref{fig:threshold}, we examine the impact of the parameter $r$ on the performance of both the RSS-based LVS and the DRSS-based LVS. We note that $r$ is the minimum distance between the claimed location and the malicious user's true location. As such, the disc determined by $\mathbf{x}_c$ and $r$ can be interpreted as the area protected by some physical boundaries. In Fig.~\ref{fig:threshold}, we  observe that the ROC curve moves toward the upper left corner as $r$ increases, which indicate that the malicious user will be easier to detect if he is further away from his claimed location. We also observe that the performance improvement due to increasing $r$ is not significant when $r$ is larger than some specific value (\emph{e.g.}, $r > 250$m).

\section{Conclusion}\label{sec_conclusion}

In this work we have formally analyzed for the first time, the performance of two important types of LVSs (RSS and DRSS-based) in the regime of spatially correlated shadowing. Our analysis illustrates that for anticipated levels of correlated shadowing both types of LVSs will have much improved performance. In addition, we formally proved that in fact a DRSS-based LVS has identical performance to that of an RSS-based LVS, for all levels of correlated shadowing. Even more surprisingly, the identical performance of RSS and DRSS-based LVSs was found to hold even when the adversary  cannot optimize his true location. We found the performance of an RSS-based LVS to be better than a DRSS-based LVS only in the case where the adversary cannot optimize \emph{all} variables under her control. The results presented here will be important for a wide range of practical location authentication systems deployed in support of emerging wireless network applications.

\section*{Acknowledgments}

This work was funded by The University of New South Wales and Australian Research Council Grant DP120102607.

\end{document}